\documentclass[letterpaper]{article}
\usepackage[ruled, lined,vlined,linesnumbered,commentsnumbered]{algorithm2e}
\usepackage{times}
\usepackage[reqno, ]{amsmath}
\usepackage{amssymb}
\usepackage{amsthm}
\usepackage{graphicx}
\usepackage{setspace}
\usepackage{url}
\usepackage{bm}
\usepackage{algpseudocode}

\usepackage{authblk}
\usepackage{subcaption}

\usepackage{aaai17} %aaai17 must be used at last
\newcommand{\ceil}[1]{\lceil #1\rceil}

\newtheorem{prop}{Proposition}
\newtheorem{theorem}{Theorem}
\newtheorem{fact}{Fact}

\algrenewcommand{\Require}{\State \textbf{Input: }}
\algrenewcommand{\Ensure}{\State \textbf{Output: }}

\makeatletter
\newcommand{\rmnum}[1]{\romannumeral #1}
\newcommand{\Rmnum}[1]{\expandafter\@slowromancap\romannumeral #1@}
\makeatother

\title{Efficient Delivery Policy to Minimize User Traffic Consumption in Guaranteed Advertising\thanks{This work was supported in part by the National Natural Science Foundation of China Grant 61222202, 61433014, 61502449, 61602440 and the China National Program for support of Top-notch Young Professionals.}}

\author[1,2]{Jia Zhang}
\author[2]{Zheng Wang}
\author[1,2]{Qian Li}
\author[1,2]{Jialin Zhang}
\author[1,2]{Yanyan Lan}
\author[1,2]{Qiang Li}
\author[1,2]{Xiaoming Sun}
\affil[1]{CAS Key Lab of Network Data Science and Technology, Institute of Computing Technology, Chinese Academy of Sciences, China.}
\affil[2]{University of Chinese Academy of Sciences, China.}
\affil[1,2]{\{zhangjia, liqian, zhangjialin, liqiang, lanyanyan, sunxiaoming\}@ict.ac.cn}
\affil[3]{Department of Computer Science, The University of Hong Kong, Hong Kong.}
\affil[3]{zwang@cs.hku.hk}

\begin{document}
\maketitle
\begin{abstract}
In this work, we study the guaranteed delivery model which is widely used in online display advertising. In the guaranteed delivery scenario, ad exposures (which are also called impressions in some works) to users are guaranteed by contracts signed in advance between advertisers and publishers. A crucial problem for the advertising platform is how to fully utilize the valuable user traffic to generate as much as possible revenue.

Different from previous works which usually minimize the penalty of unsatisfied contracts and some other cost (e.g. representativeness), we propose the novel consumption minimization model, in which the primary objective is to minimize the user traffic consumed to satisfy all contracts. Under this model, we develop a near optimal method to deliver ads for users. The main advantage of our method lies in that it consumes nearly as least as possible user traffic to satisfy all contracts, therefore more contracts can be accepted to produce more revenue. It also enables the publishers to estimate how much user traffic is redundant or short so that they can sell or buy this part of traffic in bulk in the exchange market. Furthermore, it is robust with regard to priori knowledge of user type distribution. Finally, the simulation shows that our method outperforms the traditional state-of-the-art methods.
\end{abstract}
%\keywords{online display advertising; guaranteed delivery; flow based policy}

\section{Introduction}
Online advertising is now an industry which is worth tens of billions of dollars. According to the latest report released by Interactive Advertising Bureau \cite{IABreport2015}, the annual revenue of online advertising in the US reaches \$59.6 billion in 2015, which is \$10.1 billion (or 20.4\%) higher than that in 2014. In online advertising, ad exposures (also called impressions) are sold in two major ways: either in spot through auction (real-time bidding, or RTB) or in advanced by guaranteed contracts (guarantee delivery, or GD) \cite{Chen2014dpm}. Although RTB is developing rapidly, only 13\% of publishers' revenue is generated by RTB by the year 2012 \cite{Emarketer2012w}. Large amounts of ad exposures are still sold through GD.

GD is usually adopted by top tier advertisers (Apple, P\&G, Coca-Cola etc.) and publishers (Yahoo!, Amazon etc.). Generally, top tier publishers have frequently visited web pages and large amounts of first-hand user data. These resources are valuable, and undoubtedly, heavily attractive for advertisers, especially for those top tier advertisers.
When promoting their products or brands, top tier advertisers usually prefer to sign a contract with publisher of high quality to ensure advertising effects. This contract contains targeted user types (what kinds of users the ads can only be displayed to) and the amount of exposures of corresponding ad to those targeted users.
Publisher should comply with this contract and is paid when this contract is satisfied. %Compared to RTB, there are several superiorities of GD, such as provides a safe way to get enough high quality exposures for advertisers and ensure that publishers' high quality user traffic and web pages can get a proportionate pay back. %These advantages are the key reasons why top tier publishers and advertisers are more in favor of GD than RTB.

Previous works in GD domain usually focus on meeting the requirements of all contracts. Specifically, they perform study on designing delivery policy to minimize some loss functions, which are related to unsatisfied contracts and other particular requirements such as representativeness and smoothness of delivery \cite{Chen2012Ad,Lefebvre2012SHALE}. In this work we propose a novel optimization model, in which the primary target is to minimize the user traffic\footnote{The page views of users on the publish platform.} within the constraint of satisfying all exposure requirements. Compared to previous model, this model has more practical value:
\begin{itemize}
	\item For most publishers, the \emph{redundant user traffic}\footnote{The part of user traffic that is not targeted by any contract or corresponding contracts are already satisfied.} will be sold in exchange market to avoid waste. However, compared with contract orders, selling user traffic in exchange market is usually less profitable. User traffic consumption for all contracts is minimized implies that more contracts can be accepted. Consequently, the proportion of user traffic sold to exchange market is decreased. %This is certainly a good news for publishers.
	\item Under the consumption minimization model, the amount of required user traffic can be estimated when the delivery policy is properly designed. This estimation can be viewed as an important indicator for publishers to predict how much user traffic is redundant or short. Then, publishers can sell or purchase user traffic in bulk, which is more profitable than doing that in scattered.
	\item The estimation for user traffic consumption is also useful for order management. When a particular contract comes, publisher can make decision on whether to accept it based on the change of user traffic consumption.
\end{itemize}

\subsection{Contributions}
In this work we propose a practical consumption minimization model for Guarantee Delivery. %The input of this model consists of three parts: the demand-supply graph (indicating the target relationship between ads and user types), the amount of exposures of each ad and the distribution of user types. 
Under this model, we provide the theoretical analysis and design a near optimal delivery policy. Our contributions mainly consist of two parts which are listed as follows.

\textit{Theoretical analysis for optimal expected user traffic consumption.}\ \ %According to the definition, $E(OPT)$ is the minimum expected user traffic consumption when users visiting the publish platform online and the publisher is allowed to modify already displayed ads.
By employing some fundamental results in stopping theory, we prove that the optimal offline expected user traffic consumption is lower bounded by an expected flow (\textsc{Theorem} \ref{theo:zoffline}). This is shown to be very useful for performance measurement of delivery policies. Besides, we construct an optimal online delivery policy based on a recursion. However, solving this recursion needs exponential time, which motivates us to design more efficient approximation delivery policy.

%\textit{Analysis for optimal online delivery policy.}\ \ The offline optimal $E(OPT)$ is usually impossible to achieve in the online setting as the publisher is allowed to modify previous delivery results. We then try to analysis the online optimal delivery policy and prove that it can be constructed based on a recursion. However, solving this recursion need exponential time. Nevertheless, we can show a lower bound of expected user traffic needed by the optimal online delivery policy (\textsc{Theorem} \ref{theo:onlineoptiaml}).

\textit{Efficient near optimal delivery policy.}\ \ Inspired by some insights derived from the proof of lower bound for optimal offline user traffic consumption, we design a flow based delivery policy (\textsc{Algorithm} \ref{alg:exptectedmatching}) in the online setting. This delivery policy has several advantages:
\begin{description}%[noitemsep]
	\item [(\rmnum{1})]It consumes nearly as least as possible user traffic to satisfy all exposure requirements. Let $n$ stand for the number of user types. If the total amount of all ads exposures is large enough, which is a realistic assumption in practice, the delivery policy achieves a near optimal competitive ratio $1+\frac{1}{\text{poly}(n)}$ (\textsc{Theorem} \ref{theo:2}).
		%Recall that $M$ is the total amount of exposures and $p_{\min}$ is the minimum reach probability of user types. When $Mp_{\min}=\Omega(n^{\varepsilon})$ for any $\varepsilon>0$, the expected user traffic consumed the flow based delivery policy is no more than $\left(1+O(n^{-\varepsilon/2})\right)Z_{flow}$, which means a competitive ratio $1+O(n^{-\varepsilon/2})$ is guaranteed (\textsc{Theorem} \ref{theo:2}). As $M$ is extremely large in common, $Mp_{\min}=\Omega(n^{\varepsilon})$ is a realistic constraint. Thus our delivery policy are practical in real world and achieves near optimal performance with theoretical guarantee.
	\item [(\rmnum{2})]It is robust on the priori knowledge of user type distribution. As this distribution is learned from history data, it usually has a bias compared to the real one. It is a concern that whether this bias will influence the delivery efficiency significantly. For this issue, we provide a robustness result indicating that when the bias between the learned user type distribution and the ground truth is bounded, the loss on delivery efficiency is also bounded (\textsc{Theorem} \ref{theo:robustness}).
	\item [(\rmnum{3})]Besides ad delivery, our policy can estimate user traffic consumption as well (\textsc{Theorem} \ref{theo:2}).
		%It provides reliable guidelines for order and traffic management.User traffic consumption can be easily estimated in the flow based delivery policy.
		As aforementioned, this estimation is very useful for order and traffic management.%According to \textbf{(\rmnum{1})}, $Z_{flow}$ can be taken as an prediction of user traffic consumption with negligible relative error. This is helpful for deciding whether an order can be satisfied.  %\footnote{Just calculating the new $Z_{flow}$ value assuming this order is accepted and comparing it with the maximum user traffic supply.}.
%Besides, it is indicated that how many users need of each type by the flow based delivery policy, thus, providers can find out the redundant traffic and deal with it separately, such as selling it in exchange market.
	\item [(\rmnum{4})] Our delivery policy is also flexible. It can be easily implemented to reach some other requirements, such as representativeness and smoothness of delivery (\textsc{Remark} 1.).
\end{description}
%In addition to the performance guarantee in expectation, the flow based delivery policy consumes no more than $(1+O(n^{-\varepsilon/2}))Z_{flow}$ user traffic to satisfy all contracts for almost all user visiting sequences in fact (\textsc{Remark} 1).

%In addition to above advantages, the flow based delivery policy works efficiently in real system - after an efficient offline preprocessing, each ad can be delivered within linear time. 
According to simulation experiments, our delivery policy works well under different settings and outperforms other delivery policies including the well known HWM policy \cite{Chen2012Ad}.

%As aforementioned, estimating the expected user traffic consumption can be extremely useful for order management. % When deciding whether to accept an ad campaign, the publishers can use this estimation as an important reference.
%Particularly, when this estimation is accurate enough, the publishers can calculate the marginal cost of accepting an ad campaign with enough precision, and then use this cost to guard the decision making and pricing mechanism design. Indeed, $Z_{flow}$ is a near tight estimation for our flow based delivery policy when {\small$Mp_{\min}=\Omega(n^{-\varepsilon})$}. It is safe to use $Z_{flow}$ as a reference in order management with enough guarantee.
\subsection{Related Works}
Online advertising is a hot topic concerned by researchers in both economic and computer science communities. We discuss the line of works that are most related to our paper.

The GD problem has been well studied as a special version of online assignment in theoretical domain. The classical online assignment problem studies adversary (or worst case) setting, such as \cite{Karp1990An,Mehta2007AdWords}. Some generalizations take distribution information into account. In the work \cite{Feldman2009OSM}, the authors study the online matching problem when types of users are drawn independently from a known distribution, which is similar to our setting. However, their target is to maximize the matching size when there are $n$ users in total. They propose an algorithm with competitive ratio 0.6702. This result is improved by works such as \cite{Bahmani2010,Jaillet2013}. Correspondingly, authors in the work \cite{Karande2011} study the online matching problem when the distribution of user types is unknown. In the work \cite{Vee2010Optimal}, the authors propose the online allocation problem with forecast and suppose there is a sample set of vertices that will come online. They provide a general framework to optimize several object functions which satisfy some constraints. According to \cite{Chen2012Ad}, this framework is sensitive to the accuracy of the forecast and requires solving a non-linear convex problem. We refer to the survey \cite{Mehta2012Online} for a systemic view.

Some more practical works have also been done in this field. In the work \cite{Chen2012Ad}, the authors propose a delivery policy which is called High Water Mark (or HWM). It is efficient, space saving and stateless so that can run in parallel on distributed machines. However, it is not designed to optimize the user traffic consumption. Numerical results in the experimental part of our work show that HWM is not good enough when considering user traffic consumption. In the work \cite{Lefebvre2012SHALE}, the authors design a policy called SHALE, which can be treated as an improvement of HWM. Actually, SHALE achieves a better performance on delivery with a loss on efficiency. However, SHALE is a heuristic algorithm, and usually needs specific number of iterations, which is sensitive to input, to get good enough results. A common problem for HWM and SHALE is that expected user traffic consumption is hard to estimated. Both HWM and SHALE take representativeness into consideration, which means ads in one campaign should be delivered somehow \emph{fairly} to users among all targeted user types, instead of only partial types. Besides, \emph{user reach} (the number of different individuals who see the ad) and \emph{delivery frequency} (the number of times a user should see the ad) are also concerned by some works \cite{Turner2014Delivering,Shen2014Robust}. %Among them, in the work \cite{Yang2010arXiv1008}, the authors discuss a multi-objective programming model for various objectives of advertisers and publishers. In the work \cite{Alaei2009OAD}, the authors study the GD problem when advertisers can be divided into two kinds: sales advertisers and brand advertisers, and they design an online algorithms that achieves near optimal performance in their model.

In this work, we assume the user type distribution $\mathcal{D}$ is known in advance and the types of users are sampled from this distribution, while in other works, such as \cite{Cetintas2013Forecasting,Cetintas2011FCU}, authors study how to estimate user traffic in a particular period. Another line of works concern the pricing mechanism and revenue maximization in GD \cite{Bharadwaj2010Pricing,Radovanovic2012Risk}.  In the work \cite{Bharadwaj2010Pricing}, two algorithms are provided to calculate the price of each guaranteed contract based on the price of each targeted user.

As RTB is adopted by more and more publishers, some works study how to combine RTB and GD to make more profit. In the work \cite{Ghosh2009bid}, publishers are considered to act as bidders for guaranteed contracts.  In works \cite{Araman2010Media,Balseiro2011Yield,Chen2014dpm}, a similar framework is considered that publishers can dynamically select which ad campaign contract to accept. In the work \cite{Chen2014dpm}, it is shown that RTB and GD can be seamlessly combined programmatically and the authors provide some results about what percentages of user traffic should be sold in GD and what the price is.

\section{Consumption Minimization Model}
In guaranteed delivery (GD) scenario, users are categorized into a set of \emph{atomic} types according to their traits, e.g., geographic, gender, age etc. \emph{Atomic} means each user visiting the publish platform can only belong to a unique type. Each advertiser has one or more ad campaigns and each campaign corresponds to a contract with publisher. The contract specifies how many times the corresponding ad should be exposed and the set of targeted user types. See Figure \ref{fig:1} as an example. The first ad campaign reserves 200 thousand exposures and all users at the age of 20 to 25 are targets of this campaign.
%Thus when a user whose type is \{Male, Age = 20\} or \{Female, UK, Age = 20\} comes to the publish platform, the ad corresponds to the first campaign can be delivered to this user.
Thus when a user who satisfies this requirement visits the platform, corresponding ad can be displayed to this user.
\begin{figure}[h!]
	\centering\includegraphics[width=2.2in]{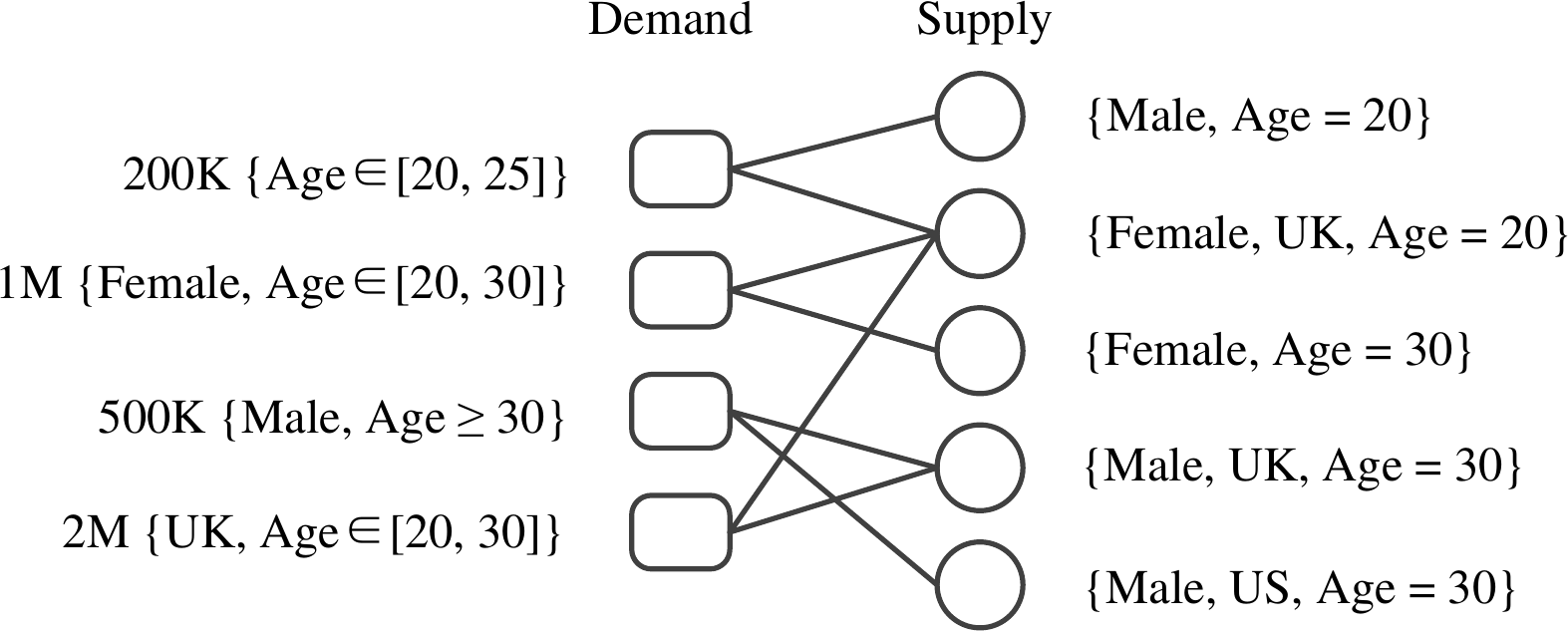}
	\caption{\label{fig:1} An example of ads campaigns in GD. %The demand side stands for the ad campaigns and the supply side is all the types of users who visit the publish platform.
}
\end{figure}

Formally speaking, the ad campaigns and user types relationship can be modeled as a demand-supply graph $G=(A,U,E)$ with $|A|=m$ and $|U|=n$. The nodes in $A$ stand for $m$ ad campaigns, while nodes in $U$ stand for $n$ atomic user types. For each ad campaign $a_i \in A$ and user type $u_j \in U$, the edge $(a_i,u_j)\in E$ means that $u_j$ is targeted by $a_i$. Conventionally, for any $v\in A\cup U$, we use $\Gamma(v)$ to stand for the neighbors of $v$ in $G$. $W_i$ stands for total amount of exposures that $a_i$ reserves in the contract. We denote $W=\{W_1, W_2,\cdots,W_m\}$ for convenience, and let $M\triangleq \sum_{i=1}^mW_i$. Under the GD setting, we can safely assume $M\gg n$.

We assume the users come one by one online. Once a user visits the publish platform, her type is revealed and the publisher should irrevocably deliver only one ad\footnote{Our model and policy can be easy extended to meet multiple delivery requirement. See supplemental material for detail.}.
Usually, there are large amounts of users' visiting logs available on the publish platform, thus, we assume that each user type is drawn from a distribution $\mathcal{D} = \{p_j\,|\,1\leq j\leq n\}$ which can be learned from those log data. That's to say, for any unseen user, she belongs to type $u_j$ with probability $p_j$, which satisfies $\sum_{j=1}^np_j = 1$. We define $p_{\min}$ (resp. $p_{\max}$) as the minimum (resp. maximum) arrival probability. $\mathcal{D}$ is considered as a part of the input. Thus, the input instance $I$ is a triple $(G,W,\mathcal{D})$.

Our objective is to design an ads delivery policy to minimize the expectation of user traffic consumed. Here the \emph{user traffic consumed} means the number of users' visits when all contracts are fulfilled. It is a random variable due to the uncertainty of user traffic and the delivery policy. This is the \emph{consumption minimization problem for user traffic} (abbreviated as CMPUT).
For a delivery policy $\pi$, we employ the concept of competitive ratio to measure its performance. We say $\pi$ has a \emph{competitive ratio} $c$ if and only if for any input instance, $\frac{E(\pi)}{E(OPT)}\leq c$, where $E(\pi)$ is the expected user traffic consumed by policy $\pi$ and $E(OPT)$ is the expectation of optimal user traffic consumption when the publisher is allowed to modify previous delivery decisions. We also call $E(OPT)$ {optimal offline expected user traffic consumption}.

We provide several important definitions for further analysis.
%\begin{defi}
	Given an input instance $I=(G,W,\mathcal{D})$ and a non-negative number $T$, we can construct an \emph{expected network} $F_G(T)$ by adding two nodes $s$ (source) and $t$ (sink) into $G$ and set the capacities as follows. For $a_i\in A$, we connect $s$ with $a_i$ and set the capacity of the edge $(s,a_i)$ as $W_i$. For $u_j\in U$, we connect $t$ with $u_j$ and set the capacity of the edge $(u_j,t)$ as $Tp_j$. For those edges in $E$, we just set their capacities as $M$.
%\end{defi}
%\begin{defi}
Let $\textsc{Max-Flow}(F_G(T))$ stand for the maximum flow that can pass through $F_G(T)$. $Z_{flow}$ is defined as the minimum value of $T$ that satisfies $\text{\textsc{Max-Flow}}(F_G(T))=M$.
%\end{defi}
%As we can see, $Z^{exp}$ is the minimum user traffic needed to deliver all $M$ ads if all types of users come as expected.

\section{Theoretical Analysis}
In this section we first prove a lower bound of the optimal offline expected user traffic consumption, and then construct an optimal online policy based on a recursion which needs exponential time to solve. To beat this hardness we develop an efficient and near optimal flow based policy.
\subsection{Lower Bound for Optimal Offline User Traffic Consumption}
\label{sec2:name}
%In this section we firstly  %As aforementioned, those lower bounds are not only useful for algorithm analysis but also useful for cost forecast and pricing mechanism design.
%We first prove a lower bound for the optimal expected user traffic for all input instance.
The main tool we employ to develop our lower bound is \emph{Wald's identity} \cite{wald1944cumulative}. Wald's identity is a foot stone in optimal stopping theory. Here we use a variant introduced in work \cite{Blackwell1946On}, and it can be described as follow.

{\bf Wald's Identity.} Let $X_1, X_2,\cdots$ be i.i.d. random variables with $E(X_1)<\infty$, and let $\tau$ be a stopping rule which is independent with $X_{\tau+1}, X_{\tau+2},\cdots$ and satisfies $E(\tau)<\infty$, then there has
$E(\sum_{i=1}^{\tau}X_i)=E(\tau)E(X_1).$

To use Wald's Identity, we need an upper bound of $E(OPT)$. Consider a simple delivery policy called \textsc{Random}, in which each ad campaign $a_i$ is expanded as $W_i$ ads. These ads target the same types of users and each of them needs to be exposed only once. When a user that belongs to $u_j$ comes, \textsc{Random} uniformly chooses one ad from all valid ones waiting for exposure. For each $u_j$, denote {$W(u_j)=\sum_{a_i\in \Gamma(u_j)}W_i$}, we have the following result.
\begin{prop}
	{\small
	\label{lemm:optupperbound}
	$$E(\text{\textsc{Random}})\leq \int_{0}^{\infty}(1-\prod_{i=1}^m(1-e^{-t\sum_{u_j\in \Gamma(a_i)}\frac{p_j}{W(u_j)}})^{W_i})dt.$$
	}
\end{prop}
\begin{proof} In \textsc{Random}, each expanded ad will be delivered with probability $\sum_{u_j\in \Gamma(a_i)}\frac{p_j}{W(u_j)}$ when the  first user comes. This probability will increase with more and more ads are delivered. If we fix the probability just as $\sum_{u_j\in \Gamma(a_i)}\frac{p_j}{W(u_j)}$, the GD problem is reduced to the \emph{non-uniform coupon collection problem}, and according to the work \cite{flajolet1992birthday}, the expected user traffic consumption can be calculated by an integral as shown in this proposition. Clearly, this integral will converge to a value, which is an upper bound of $E(\textsc{Random})$.
\end{proof}

According to above proposition, $E(OPT)\leq E(\textsc{Random})<\infty$.

Denote $U_j$ as the number of users belong to type $u_j$ when all contracts are fulfilled and let $[n]$ stand for the set $\{1,2,\cdots,n\}$. We can use Wald's Identity to show the following proposition.
\begin{prop}\label{lemm:stopping}
	$E(U_j)=p_jE(OPT)$ for any $j\in[n]$.
\end{prop}
\begin{proof}
This is a direct application of Wald's Identity.
Fix any $j\in [n]$, define
\begin{equation*}
	X_t=\left\{
	\begin{aligned}
		&1,& &\text{the $t$-th user belongs to type $u_j$},&\\
		&0,& &\text{otherwise}.&
	\end{aligned}\right.
\end{equation*}
 We can find that $\{X_t\}_{t\geq 1}$ is a sequence of i.i.d. random variables and $E(X_t)=p_i<\infty$. In addition, the stopping time  of the policy, i.e. $OPT$ is independent of $X_{OPT+1},X_{OPT+2},\cdots$ and $E(OPT)<\infty$. Thus $E(U_j)=E\left(\sum_{t=1}^{OPT}X_t\right)=p_jE(OPT).$
\end{proof}
\begin{theorem}
	\label{theo:zoffline}
For any input instance, $E(OPT)\geq Z_{flow}$.
\end{theorem}
\begin{proof}
	According to the definition of $Z_{flow}$ and the property of maximum flow, there exists a subset $S\subseteq A$ such that
		$\sum_{a_i\in S}W_i\geq \sum_{u_j\in \Gamma(S)}p_jZ_{flow},$
	where $\Gamma(S)\subseteq U$ is the set of user types targeted by ad campaigns in $S$. Otherwise, $Z_{flow}$ can be decreased to a smaller value which still satisfies $\textsc{Max-Flow}(F_G(Z_{flow}))=M.$ It is a contradiction to the definition of $Z_{flow}$.

On the other hand, the delivery procedure can stop only if the ad corresponding to campaign $a_i$ has been delivered $W_i$ times, thus $\sum_{u_j\in \Gamma(S)} U_j$ is always no less than $\sum_{a_i\in S}W_i$.

According to \textsc{Proposition} \ref{lemm:stopping}, we have $E(U_j)=p_jE(OPT)$ for any $i\in[n]$, thus
{\small
\begin{equation*}
	\sum_{u_j\in \Gamma(S)}p_jE(OPT)=E\sum_{u_j\in \Gamma(S)} U_j\geq \sum_{u_j\in \Gamma(S)}p_jZ_{flow},
\end{equation*}
}which implies $E(OPT)\geq Z_{flow}$.
\end{proof}

\subsection{Optimal Online Delivery Policy}
%In this section, we focus on the optimal online delivery policy.
Suppose the first user comes with type $u_j$. If no ad campaign targets $u_j$, nothing will be delivered to this user; otherwise, the delivery policy should decide which ad in $\Gamma(u_j)$ will be delivered. It is easy to argue that if  $\Gamma(u_j)$ is not empty, delivering nothing can not be better than delivering any arbitrary ad to the user. Suppose $a_i\in \Gamma(u_j)$ is delivered, then the problem is reduced to a sub problem in which ad $a_i$ should be delivered $W_i-1$ times.  For convenience, we construct an input $I(i)=(G,W(i),\mathcal{D})$, where $W(i)=\{W_1,\dots,W_{i-1},\max\{0, W_i-1\},W_{i+1},\dots,W_m\}.$ Based on above analysis, in the optimal online delivery policy $\pi^*$, $a_i$ must be chosen to satisfy that
%{\begin{equation*}
	$a_i = \arg\min_{a_k\in \Gamma(u_j)}E\left(\pi^*(I(k))\right)$,
%\end{equation*}}
where $E(\pi^*(I(k)))$ stands for expected user traffic consumed by $\pi^*$ on input instance $I(k)$.
According to the properties of conditional expectation, we have
	$E(\pi^*(I)) = 1+\sum_{j=1}^np_j\min_{a_i\in \Gamma(u_j)}E(\pi^*(I(i)))$.
Thus, $\pi^*$ can be constructed from bottom to top by dynamic programming (DP).
However, it is easy to check that there are $\prod_{i=1}^m (W_i+1)$ sub problems in this DP. Thus, exponential time is needed to construct $\pi^*$ in this way. In fact, we believed that $\pi^*$ is hard to construct in any way. This motivates us to design some delivery policies with polynomial running time. %However, this hardness doesn't affect us to provide a lower bound for the optimal online delivery policy.

\subsection{Delivery Policy Design for CMPUT}
\label{sec3:name}
According to previous analysis, $Z_{flow}$ is highly related to $E(OPT)$. It somehow reminds us to use expected flow to guide ads delivery. Inspired by this insight, we design a flow based delivery policy.

\subsubsection{Flow Based Delivery Policy}
As shown in \textsc{Algorithm} \ref{alg:exptectedmatching}, the flow based delivery policy can be divided into two parts: the offline part and the online part. In the offline part, we first construct an expected network with integer capacities on all edges, that's $\hat{F}_G(\hat{Z})$ used in \textsc{Algorithm} \ref{alg:exptectedmatching}. The difference between $F_G(\hat{Z})$ and $\hat{F}_G(\hat{Z})$ lies in that the capacity of edge $(u_j, t)$ is set to be $\ceil{\hat{Z}p_j}$ for $j\in [n]$ in the latter. Then we find the minimum value of $\hat{Z}$ while ensuring the max flow on $\hat{F}_G(\hat{Z})$ is no less than $M$.
To find such a value of $\hat{Z}$, we extend the idea of binary search, that's to say, we enumerate a value $\hat{Z}$, and calculate $\textsc{Max-Flow}(\hat{F}_G(\hat{Z}))$, then adjust the value of $\hat{Z}$ with a similar idea of binary search. According to the properties of maximum flow, the traffic on each edge in maximum flow is a non-negative integer when the capacities of all edges are integers. This is why we use $\hat{F}_G(\hat{Z})$ instead of $F_G(\hat{Z})$.
In the online part, when a user comes, we deliver a proper ad according to the maximum flow of $\hat{F}_G(\hat{Z})$. %As we can see, our allocation policy runs efficiently: after a offline calculation in $O(XXX)$ time, each allocation can be decide in $O(1)$ time.

Note that we haven't specified how to deliver an ad according to the maximum flow of {$\hat{F}_G(\hat{Z})$} in the line \ref{line:alloc}. Actually, \textsc{Flow-Based-Delivery-Rule} stands for a class of delivery rules which satisfy that
when there are no less than {$\ceil{\hat{Z}p_j}$} users of type $u_j$ come to the publish platform, either all ads of $a_i$ are delivered or no less than {$C_{i,j}$} ads of $a_i$ are delivered to users of type $u_j$.
%\end{enumerate}
We can implement a \textsc{Flow-Based-Delivery-Rule} according to particular requirements, such as representativeness and smoothness of delivery. %See Remark 1 and 2 for detail.
Here we give a simple implementation in \textsc{Algorithm} \ref{alg:greedy}.

Next, we will prove that the performance of \textsc{Algorithm} \ref{alg:exptectedmatching} is quite good by showing that it is close to the lower bound provided in \textsc{Theorem} \ref{theo:zoffline} when $Mp_{\min}=\Omega(n^{-\varepsilon})$ for $\varepsilon>0$. The main tool we employ is  the Chernoff Bound, which is an important tool in probability theory. % and can be depicted as follow.
Besides the Chernoff Bound, the full proof of \textsc{Theorem} \ref{theo:2} mainly depends on delicate analysis of expectation. We leave this proof in supplemental material due to space limitation.
%{\bf Chernoff Bound.} Suppose $X_1,\cdots, X_n$ are i.i.d random variables taking values in $\{0,1\}$. Let $X$ denote their sum and $\mu = E(X)$, then for any $\delta \in (0,1)$,
%\begin{equation*}
%	$E(X\leq (1-\delta)\mu) \leq e^{-\frac{\delta^2\mu}{2}}$.%\left(\frac{e^{-\delta}}{(1-\delta)^{1-\delta}}\right)^{\mu}.
%\end{equation*}

\IncMargin{1em}
\SetAlFnt{\small}
\SetAlCapFnt{\small}
\SetAlCapNameFnt{\small}
\begin{algorithm}[h!]
	\caption{\label{alg:exptectedmatching}Expected Flow Based Delivery Policy}
	\Indm
	\textsc{Input:} $I=(G,W,\mathcal{D})$.\\
	\textsc{Offline part:}\\
	\Indp
	Set $\hat{Z}\gets M/2.$\\
	\Repeat{$\textsc{Max-Flow}(\hat{F}_G(\hat{N}))\geq M$}{
	Set $\hat{Z}\gets 2\hat{Z}.$\\
	Construct an expected network $\hat{F}_G(\hat{Z})$.\\% with integer edge capacities.\\
	}
	Set $Begin\gets \hat{Z}/2$, $End \gets \hat{Z}$.\\
	\Repeat{$Begin\geq End$}{
		Set $Mid\gets (Begin+End)/2$.\\
		\If{$\textsc{Max-Flow}(\hat{F}_G(\hat{N})) < M$}
		{
			Set $Begin\gets Mid+1$.\\
		}
		\Else{
			Set $End\gets Mid$.\\
		}
	}
	Set $\hat{Z}\gets Begin$.\\
	\For{each $(a_i,u_j)\in E$}{
	Set $C_{i,j}$ as the traffic passing through $(a_i,u_j)$ in the maximum flow of $\hat{F}_G(\hat{Z})$.
	}
	\Indm
	\textsc{Online part:}\\
	\Indp
	\For{each user comes to the publish platform}
	{

		\textbf{if }$M=0$ \textbf{then\ \ exit}.\\
		Suppose $u_j$ is the type of this user.\\
		\lnlset{line:alloc}{REM}Set $i\gets$ \textsc{Flow-Based-Delivery-Rule}$(u_j)$.\\
		\lnlset{line:if}{IF\ }\If{$i=NULL$}
		{Taking this user for other purpose.\\}
		Deliver an ad of $a_i$ to current user.\\
		Set $W_i\gets W_i-1$ and $M\gets M-1$.\\
	}
\end{algorithm}
\begin{algorithm}[h!]
	\caption{\label{alg:greedy}\textsc{Greedy-Delivery-Rule}}
	\Indm
	\textsc{Greedy-Delivery-Rule}$(u_j)$\\
	\Indp
	Set $i'\gets \arg\max_{a_i\in \Gamma(u_j), W_j>0}\left\{C_{i,j}\right\}$.\\
	\If{we can't find such an $i'$}
	{
		\textbf{return} $NULL$.
	}
	Set $C_{i',j}\gets C_{i',j}-1$.\\
	\textbf{return} $i'$.
\end{algorithm}

\newcommand{\theoperformance}{Given $\varepsilon>0$, when $Mp_{\min}=\Omega\left(n^{\varepsilon}\right)$, the expected user traffic consumption of \textsc{Algorithm} \ref{alg:exptectedmatching} is at most $\left(1+O\left(n^{-\varepsilon/2}\right)\right)\cdot Z_{flow}$, which implies that it achieves a competitive ratio $1+O\left(n^{-\varepsilon/2}\right)$.}
\begin{theorem}
	\label{theo:2}
	\theoperformance
\end{theorem}

\noindent\textbf{Remark 1.}\ It should be pointed out that the flow based delivery policy can be generalized to solve the representativeness, smoothness and multiple delivery issues. The details can be found in the supplemental material.

\noindent\textbf{Remark 2.}\ The flow based delivery policy provides a threshold that how many users of each type are needed to satisfy all contracts, thus user traffic beyond this threshold can be safely utilized for other purpose, such as sell it in an exchange market. It is shown in Line \ref{line:if} of \textsc{Algorithm} \ref{alg:exptectedmatching}.

\subsubsection{Robustness}
In our model, we assume the user type distribution is accurate. However, there usually has a bias between estimated distribution and the real one. The following analysis provides a robustness guarantee for flow based delivery policy.% when there are some inaccuracies for the distribution $\mathcal{D}$.

Given a delivery policy $\pi$, denote $E_{\mathcal{D}}(\pi(\hat{\mathcal{D}}))$ as the expected user traffic consumption when the distribution input to $\pi$ is $\hat{\mathcal{D}}$ while the real distribution is $\mathcal{D}$. We also define $\hat{p}_j$ as the estimated probability for $p_j$ in $\hat{\mathcal{D}}$. Now we have the following theorem for the flow based delivery policy. The full proof is moved to supplemental material for space limitation.

\newcommand{\theorobustness}{Denote the flow based delivery policy as $\pi_f$, for any $\delta > 0$, if $(1-\delta)p_j \leq \hat{p}_j\leq (1+\delta)p_j$ for each user type $a_j$, we have $E_{\mathcal{D}}(\pi_f(\hat{\mathcal{D}}))\leq \frac{(1+O(n^{-\varepsilon/2}))(1+\delta)}{1-\delta}E_{{\mathcal{D}}}(\pi_f({\mathcal{D}})).$}
\begin{theorem}
	\label{theo:robustness}
    \theorobustness
\end{theorem}

\section{Experiments }
In this section, we perform some simulations of the flow based delivery policy and several other policies on different random graphs. All experiment arguments are selected to approximate the real world system.
%Despite a gap between simulations and real world,
The results show the superiority of the flow based delivery policy.
\subsection{Experimental Setup}
\noindent\textbf{Hardware}\ \ A desktop with an Intel i5-4570 CPU @ 3.2GHZ and a 4GB DDR3 memory.

\noindent\textbf{Settings}\ \
Restricted by the business model of GD, the size of demand-supply graph in real system is usually not large. Thus we conduct experiments on a random demand and supply graph with $500$ ad campaigns and $1000$ user types. Edges are generated uniformly and randomly while the average degree of user types are fixed. To generate the user types distribution $\mathcal{D}$, we choose a random value $r_j>0$ for each $u_j$, and set $p_j = \frac{r_j}{\sum_{l=1}^{1000} r_l}$. $r_j$ is chosen in two ways: $r_j$ is drawn uniformly from $[0,1]$ (we call this method Random-Normalization) or $r_j$ equals to $\frac{1}{1000}+\varepsilon$, where $\varepsilon\sim \mathcal{N}(0, 1/6000^2)$ (we call this method Gauss-Perturbation).

\noindent\textbf{Delivery Policies}\ \
We compare the flow based delivery policy (FB) with other four delivery policies, including \textsc{Random}, HWM, and two intuitive greedy policies: \textsc{Probability-Greedy} (PG) and \textsc{Degree-Greedy} (DG).

Policies DG and PG are based on a straightforward idea that ad campaigns that are \emph{difficult} to be fulfilled should have a higher priority. Specifically, when a user comes, the ad with minimum degree on demand-supply graph is chosen in DG, while in PG, it calculates a \emph{delivery value} $r_i = \sum_{u_j\in \Gamma(a_i)}\frac{p_j}{W(u_j)}$ for each ad campaign $a_i$, and delivers an ad with minimum delivery value. Recall that $W(u_j) = \sum_{a_i\in \Gamma(u_j)}W_i$.
For these policies, we compare their competitive ratios and the robustness on different random graphs and different user visiting sequences. On each generated random graph, we simultaneously run five policies with 100 sampled user sequences. The competitive ratio for each policy is estimated as dividing the average user traffic consumed on these sampled sequences by the average optimal user traffic needed.

\subsection{Experiments under Different Average Degrees}
The first experiment is conducted under different average degrees of user types. With a specific average degree, there are 100 input instances generated independently and randomly. Half of them use Random-Normalization to generate $\mathcal{D}$, the others use Normalization-Perturbation. $W_i$ is drawn uniformly from $[50,100]$. In the second experiment we will show that the amount of exposures has little influence on performance, thus a small value for $W_i$ is enough.

Experiment results for average degrees $5,10,15,20$ and $25$ are shown in Figure \ref{fig:degree:change}. The height of the bar stands for the average competitive ratios on 50 input instances with corresponding setting while the vertical lines indicate the fluctuation ranges of those competitive ratios. According to this figure, the flow based delivery policy has a much better competitive ratio under all average degrees. These fluctuation ranges of competitive ratios also show that the flow based delivery policy is more robust on different graphs. 

Surprisingly, DG is the second best policy.
When the average degree is $5$, its average competitive ratio is very close to that of the flow based delivery policy with a gap about $0.015$ for Random-Normalization and $0.035$ for Gauss-Perturbation. However, these gaps will be enlarged as the average degree increases to 10, 15 and 20. 
From Figure \ref{fig:degree:change} we can also find the fact that the performances of all delivery policies will be improved continuously with average degree increasing. In fact, when the average degree is large enough, all polices are near optimal\footnote{Considering an extreme case that when the demand-supply graph is complete, any reasonable delivery policy will be optimal.}. %Figure \ref{fig:degree:40} shows the average competitive ratios when average degree for user types is 40. All policies perform well under this setting, while our delivery policy still keeps dominant. It should be pointed out that the supply demand graph in this case is so dense that
However, for the purpose of accurately targeting, the demand-supply graph is usually very sparse in real world.
%\begin{figure}[h]
%		\centering\includegraphics[width=1.9in]{fig-degree-40}
%			\caption{\label{fig:degree:40}Comparison under average degree 40. All delivery policies achieve excellent performance, while our flow based policy still keeps dominant.}
%\end{figure}

%(see Appendix A.1 in the full version\,\footnote{The full version can be found in https://goo.gl/jHTVS6.} for detail).

%Figure \ref{fig:degree:change} shows the results in a highly summarized level. %Indeed, the flow based delivery policy shows a superiority in terms of competitive ratio and robustness on each input instance.
In Figure \ref{fig:degree:50}, we show more detail experiment results of five delivery polices on 50 input instances with Random-Normalization and average degree 10. Figure \ref{fig:avg} shows the results of competitive ratios. The horizontal axis stands for 50 input instances. As we can see, the flow based delivery policy achieves a much better competitive ratio on each input instance and it works more steadily. The worst case performance among all sampled user visiting sequences on 50 input instances is shown in Figure \ref{fig:worst}. Clearly, the flow based delivery policy is more robust on the worst sampled sequence too.

\subsection{Experiments under Different Exposure}
%In previous experiment the number of exposures for each ad campaign is at most 100, which is far less than that in real world.
In this part, we compare five polices under different exposure settings. The number of exposures for each ad campaign is uniformly drawn from 100 to 2500, 5000, 10000 and 20000. We only consider average degrees of $5$ and $10$. Similarly, 100 graphs are drawn independently, and for 50 of them, we construct distributions by Random-Normalization and for the rest, Gauss-Perturbation is used. Then input instances are constructed with different exposure settings.

\begin{figure}[h!]
	\centering
	\begin{subfigure}[b]{0.48\columnwidth}
		\centering\includegraphics[height=0.87in]{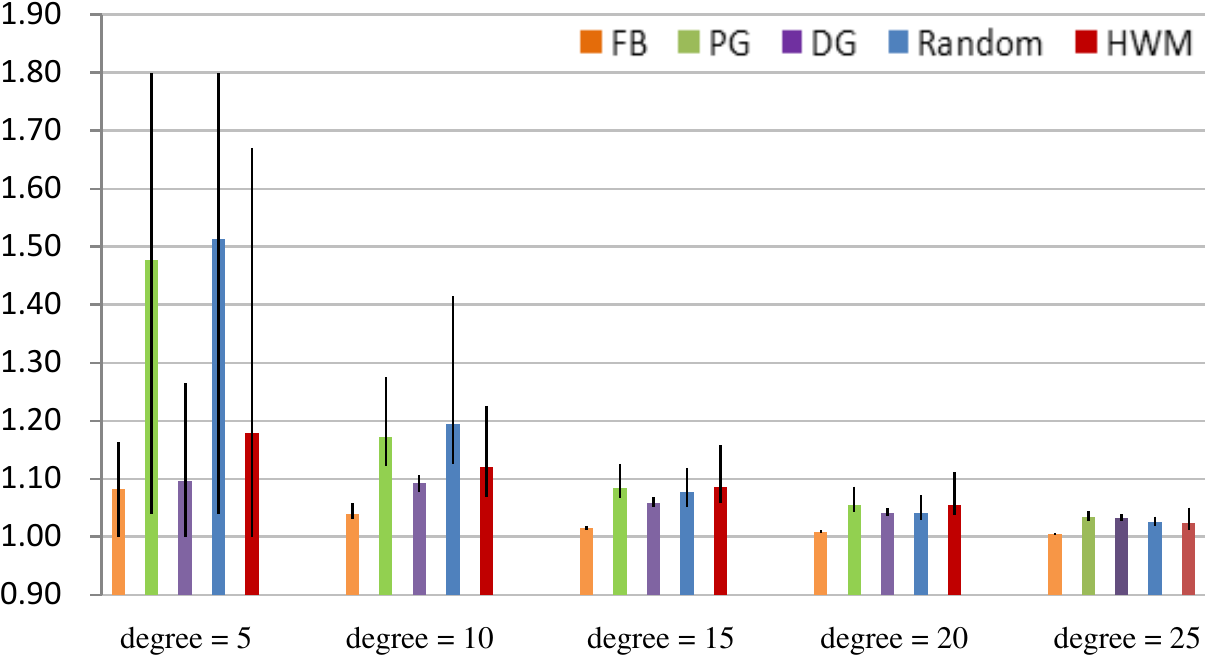}
		\caption{Random-Normalization}
	\end{subfigure}
	\begin{subfigure}[b]{0.48\columnwidth}
		\centering\includegraphics[height=0.87in]{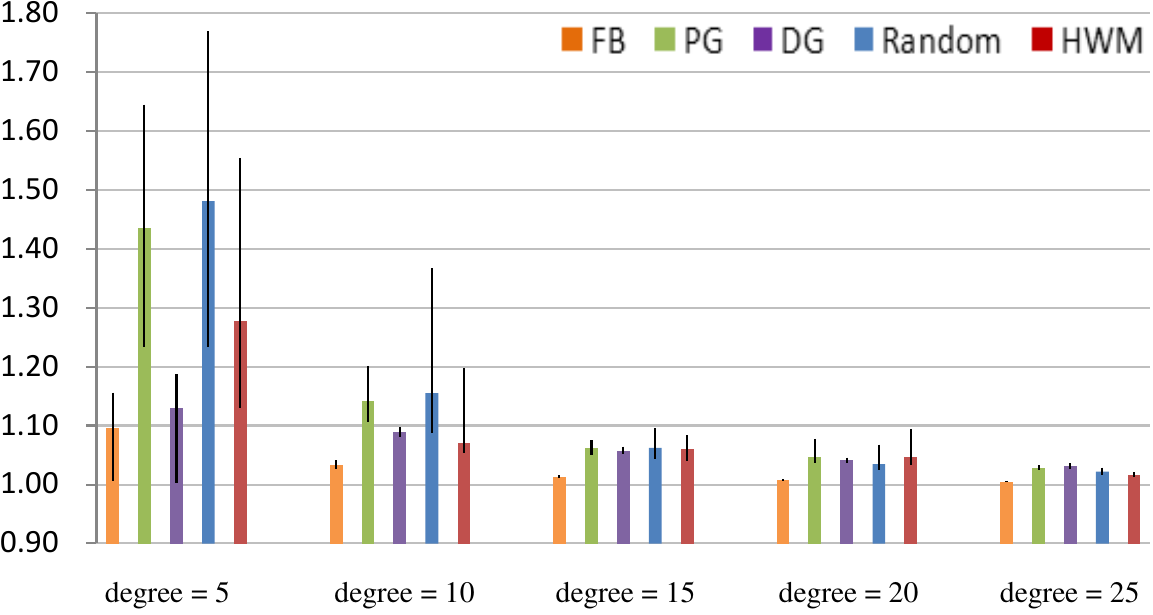}
		\caption{Gauss-Perturbation}
	\end{subfigure}
	\caption{\label{fig:degree:change}Comparison under different average degree with Random-Normalization and Gauss-Perturbation. The height of the bar is the average of competitive ratios on 50 input instances and the vertical lines indicate the fluctuation ranges.}
\end{figure}
\begin{figure}[h]
	\centering
	\begin{subfigure}[b]{\columnwidth}
		\centering\includegraphics[width=2.7in]{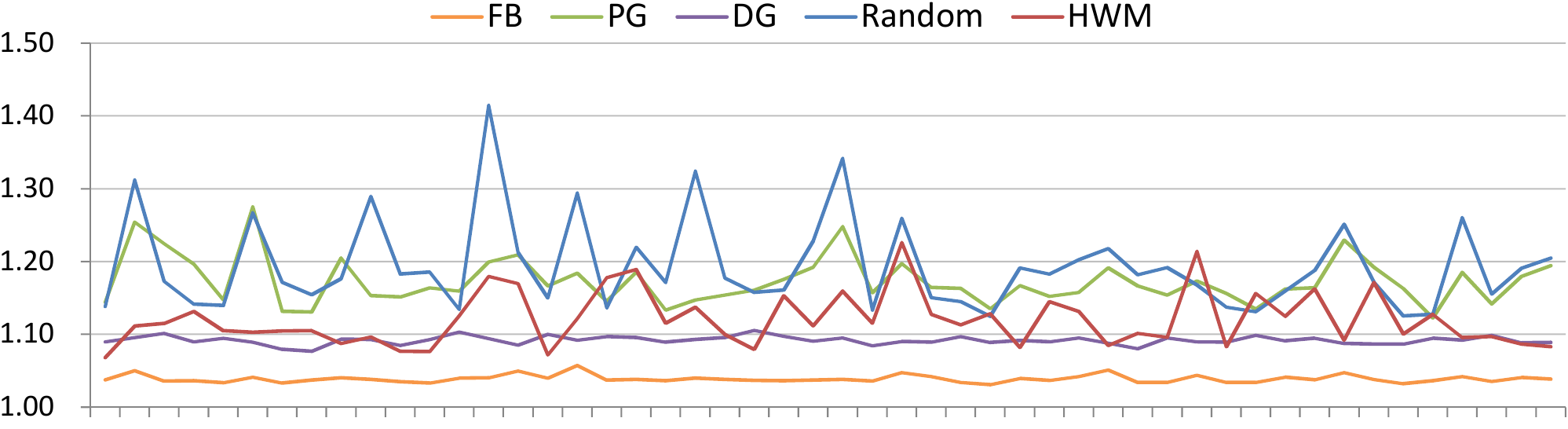}
		\caption{\label{fig:avg}Competitive ratios}
	\end{subfigure}
	\qquad\\
	\begin{subfigure}[b]{\columnwidth}
		\centering\includegraphics[width=2.7in]{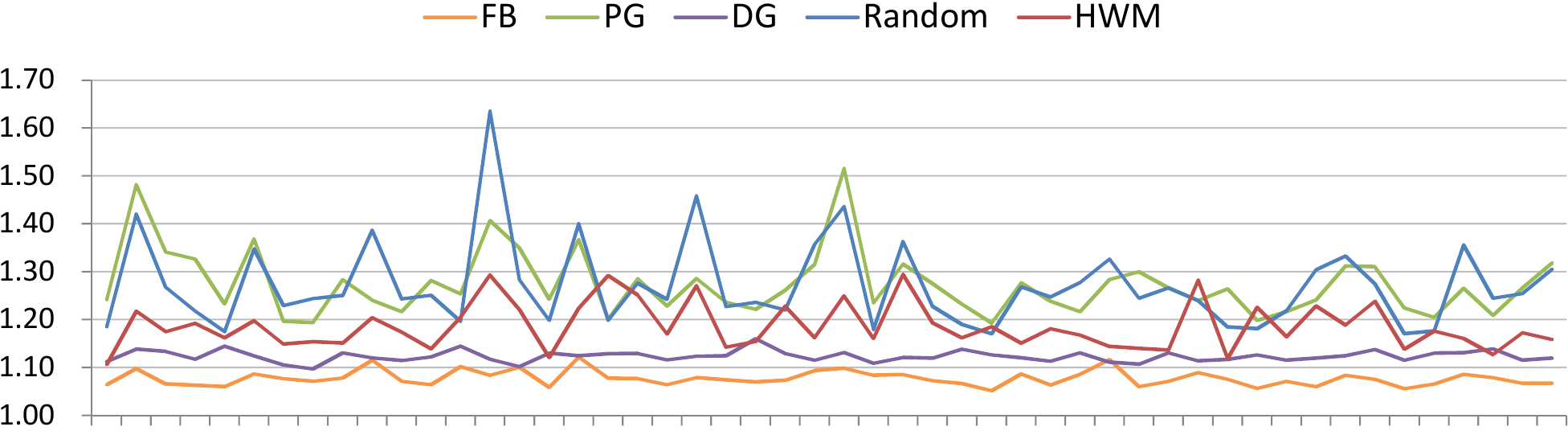}
		\caption{\label{fig:worst}Performance on the worst user visiting sequence}
	\end{subfigure}
	\caption{\label{fig:degree:50}Results for 50 input instances with average degree 10 and Random-Normalization. The horizontal axis stands for those input instances.}
\end{figure}

\begin{figure}[h!]
	\centering
	\begin{subfigure}[b]{0.49\columnwidth}
		\centering\includegraphics[height=0.84in]{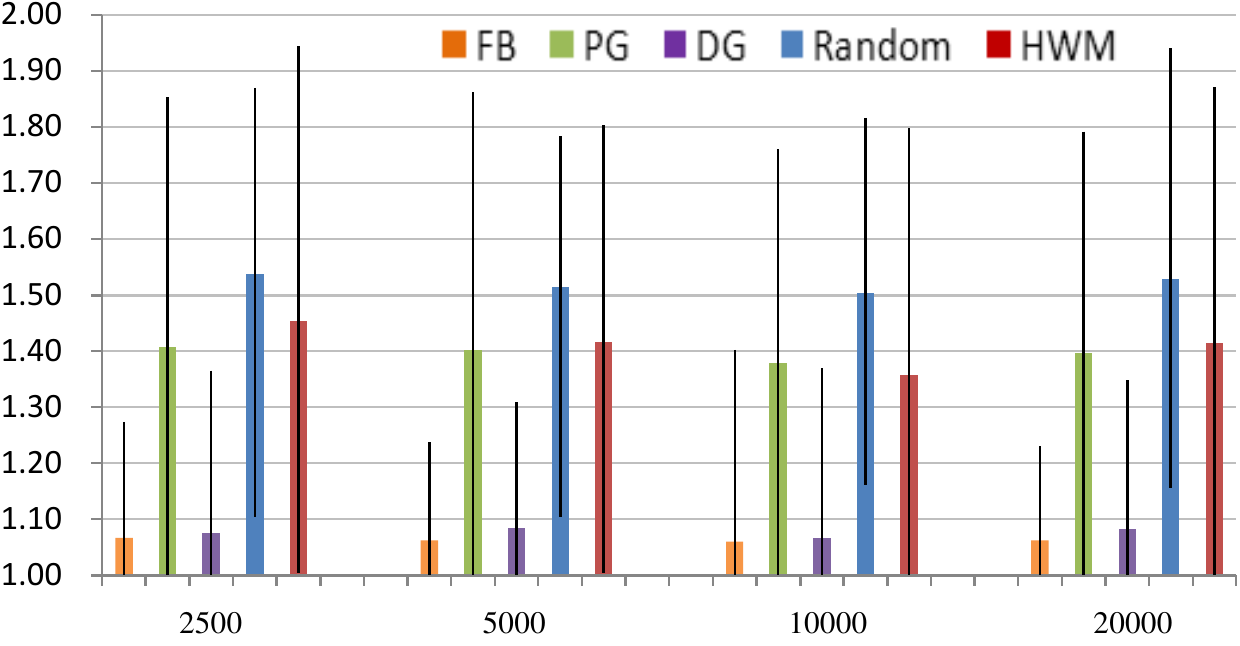}
		\caption{\label{fig:expo:random:a}degree $=5$}
	\end{subfigure}
	\begin{subfigure}[b]{0.49\columnwidth}
		\centering\includegraphics[height=0.84in]{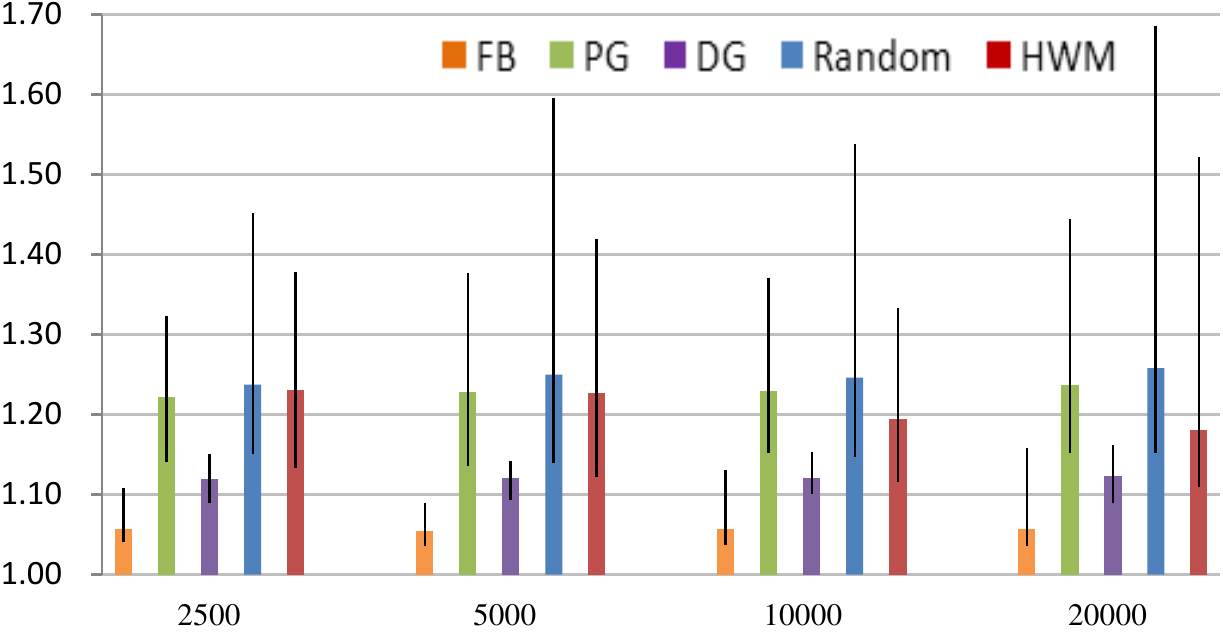}
		\caption{\label{fig:expo:random:b}degree $=10$}
	\end{subfigure}
	\caption{\label{fig:expo:random}Comparison under different exposures with Random-Perturbation.}
\end{figure}

Figure \ref{fig:expo:random} shows the results for Random-Normalization. %For Gauss-Perturbation, the results are similar and we omit it. %Please see Figure \ref{fig:expo:gauss} in Appendix \ref{app:name:2}.
The flow based delivery policy outperforms other polices under all exposure settings. %When the average degree equals $5$, though the average competitive ratio of DG is close to that of flow based delivery policy, the latter is more robust than the former.
We point out that the amount of exposures seems to have little influence on the performance for each delivery policy. It is somehow means the user traffic consumption of each studied policy has a linear relationship with the total amount of exposures for a fixed demand-supply graph. Besides, comparing Figure \ref{fig:expo:random:a} and \ref{fig:expo:random:b} and combining previous experiments, we can find that the structure of demand-supply is much more crucial for the GD problem.
Similar results are obtained for Gauss-Perturbation. These results can be found in supplemental material. 

\section{Conclusions and Future Works}
In this work, we study the consumption minimization problem for user traffic (CMPUT), which lies in the heart of GD scenario. We provide a lower bound for the optimal offline user traffic consumption by using an expected flow, which is shown to be useful for performance evaluation. We then propose a flow based delivery policy. This delivery policy consumes nearly as least as possible user traffic to satisfy all contracts under a realistic constraint. What's more important, it is robust with regard to estimation error of user type distribution and shown to be very useful for order and traffic management.

Our flow based delivery policy can be extended to solve other issues in GD, such as representativeness and smoothness of delivery. Nevertheless, there are some works remained to be done. For example, there may be a time window for each advertiser which indicates the ad campaign's starting time and deadline. Another interesting problem is how to design delivery policy when advertisers also come online instead of being revealed in advance. 
\bibliographystyle{aaai}
%\clearpage
\bibliography{ref}

\clearpage
\noindent\textbf{Supplemental Material}\\
\appendix
\section{Omitted Proofs}
\noindent\textbf{Theorem \ref{theo:2}.} \textit{\theoperformance}\\
\begin{proof}
	According to \textsc{Theorem} \ref{theo:zoffline}, we have $E(OPT)\geq Z_{flow}$. By the definitions of $Z_{flow}$ and $\hat{F}_G$, it is clear that
	\begin{equation*}
		\textsc{Max-Flow}(\hat{F}_G(Z_{flow}))\geq M.
	\end{equation*}
	Thus, $\hat{Z}\leq Z_{flow}$, as $\hat{Z}$ takes the minimum value that satisfies $\textsc{Max-Flow}(\hat{F}_G(\hat{Z}))=M$ after the execution of \textsc{Online part}.
	%If we can show that $E(M) \leq (1+\varepsilon)N$ with high probability, the theorem can be proved. %users come to the platform, all ads are allocated by Algorithm \ref{alg:exptectedmatching} with high probability. On the other hand, the expected users needed in the optimal algorithm is no less than $(1-\varepsilon)N$ with high probability.

	Let $Y$ stand for the user traffic consumption of \textsc{Algorithm} \ref{alg:exptectedmatching}. Now we establish a relationship between $E(Y)$ and $\hat{Z}$. We first bound the probability that $Y$ is greater than $(1+\varepsilon')\hat{Z}$ for a positive number $\varepsilon'$. The range of $\varepsilon'$ will be determined later. Suppose there are $(1+\varepsilon')\hat{Z}$ users come to the platform, and denote $U_{j}$ as the number of users belong to type $u_j$ among them. $U_j$ is a random variable in fact.
%Let $C'_{i,j}$ stand for the number of ads $a_i$ that allocated to user with type $u_j$ by Algorithm \ref{alg:exptectedmatching}.
According to the property of \textsc{Flow-Based-Delivery-Rule}, if $U_j\geq \ceil{\hat{Z}p_j}$, for any $a_i\in \Gamma(u_j)$, either no less than $C_{i,j}$ ads of $a_i$ is delivered to users belong to $u_j$ or all $W_i$ ads are delivered. Thus, if $U_j\geq \ceil{\hat{Z}p_j}$ is held for all $j\in [n]$, all $M$ ads are ensured to be delivered by our policy, as $\textsc{Max-Flow}(\hat{F}_G(\hat{Z}))=M$.
	
Now we bound the probability that $U_j\geq \ceil{\hat{Z}p_j}$ for $j\in [n]$. As $\hat{Z}+n\geq \sum_{j=1}^n \ceil{\hat{Z}p_j} \geq M$,
{
\begin{equation*}
	E(U_j)=(1+\varepsilon')\hat{Z}p_j > (M-n)p_{\min}=\Omega(n^{\varepsilon}),
\end{equation*}
} as $p_{\min}\leq \frac{1}{n}$. Let $\delta = \frac{\varepsilon'}{1+\varepsilon'}-\frac{1}{E(U_j)}$. If $\delta > 0$, we have
{\small
	\begin{equation*}
			\Pr(U_j< \ceil{\hat{Z}p_j}) \leq \Pr\left(U_j \leq \left(1-\delta \right)E(U_j)\right)
			%&\leq \left(\frac{e^{-\delta}}{\left(1-\delta\right)^{1-\delta}}\right)^{E(U_j)}\\
			=O\left(e^{-\delta^2E(U_j)/2}\right),%O\left(\left(\frac{e^{\delta}}{(1-\delta)^{1-\delta}}\right)^{-n^\varepsilon}\right).
	\end{equation*}
}according to Chernoff Bound. Define $\delta'=\frac{\varepsilon'}{1+\varepsilon'}-\frac{1}{(1+\varepsilon')\hat{Z}p_{\min}}\geq 0$. Clearly, $\delta'\leq \delta$.
	We have {\small
	\begin{equation}
		\label{equ:temp:union}
		\Pr(\exists u_j\in U, U_j <\ceil{\hat{Z}}p_j) \leq cne^{-\delta'^2(1+\varepsilon')\hat{Z}p_{\min}/2},
	\end{equation}
}for a proper constant number $c>0$. %with high probability $U_j\geq\ceil{\hat{Z}p_j}$ for $j\in [n]$. Consequently,

	Let $F$ be the cumulative distribution function of $Y$. The Equation (\ref{equ:temp:union}) implies
	{\small
	\begin{equation}
		\label{equ:temp:y}
		F((1+\varepsilon')\hat{Z})\geq 1-cne^{-\delta'^2(1+\varepsilon')\hat{Z}p_{\min}/2}.
	\end{equation}
}

Next, we bound the expectation of $Y$ based on the above estimation. Since $Y$ is a discrete random variable, in order to simplify our proof, we define a \emph{continuous} random variable $Y'$ which satisfies $Y'\in [M,+\infty)$ and $$F'((1+\varepsilon')\hat{Z})=1-cne^{-\delta'^2(1+\varepsilon')\hat{Z}p_{\min}/2}, (\forall \delta'>0)$$ where $F'(y)$ is the cumulative distribution function of $Y'$. It should be emphasized that $F'$ can be taken as a function of $\varepsilon'$ because $\delta'$ is a function of $\varepsilon'$. This will be useful for further analysis. Let $T=(1+\varepsilon')\hat{Z}$, we have
{\small
	\begin{equation}
		\label{equ:temp:ey}
		\begin{split}
			E(Y) & = \sum_{i = M}^{+\infty}\left(F(i)-F(i-1)\right)i\\
			&\leq T+\sum_{i=T+1}^{+\infty}\left(F(i)-F(i-1)\right)i\\
			&= T+\lim_{b\rightarrow +\infty}\left(bF(b)-(T+1)F(T)-\sum_{i =T+1}^{b-1}F(i)\right).\\
			%&\leq T+\lim_{b\rightarrow +\infty}\left(bF'(b)-(T+1)F'(T)-\sum_{i =T+1}^{b-1}F'(i)\right)\\
		\end{split}
	\end{equation}
	}Note that for any integer $y$ that satisfies $y\geq T$, we have $F(y)\geq F'(y)$ according to their definitions, thus %$F'$ can be taken as a function of $\varepsilon'$.
{\small
	\begin{equation*}
		\begin{split}
			0&\leq \lim_{b\rightarrow\infty}b(F(b)-F'(b))\leq \lim_{b\rightarrow\infty}b(1-F'(b))\\
			&= \lim_{x \rightarrow \infty}(1+x)\hat{Z}\cdot cne^{-\left(\frac{x}{1+x}-\frac{1}{(1+x)\hat{Z}p_{min}}\right)^2(1+x)\hat{Z}p_{\min}/2} = 0.
		\end{split}
	\end{equation*}
}The second line of above inequality is obtained by replacing $b$ with $(1+x)\hat{Z}$. Above inequality implies that $\lim_{b\rightarrow +\infty}bF(b)=\lim_{b\rightarrow +\infty}bF'(b)$. Applying this equation into Equation (\ref{equ:temp:ey}), we have
{\small
\begin{equation*}
		\begin{split}
			E(Y) & \leq T+\lim_{b\rightarrow +\infty}\left(bF'(b)-(T+1)F'(T)-\sum_{i =T+1}^{b-1}F'(i)\right)\\
	&= T+\sum_{i = T+1}^{+\infty} (F'(i)-F'(i-1))i \\
	&\leq T+ \int_{T}^{+\infty} xdF'(x)\\
	&= T+\int_{\varepsilon'}^{+\infty} (1+x)\hat{Z}\, d\,F'((1+x)\hat{Z})\\
	&=T+\int_{\varepsilon'}^{+\infty}(1+x)\hat{Z}\, d\, cne^{-\left(\frac{x}{1+x}-\frac{1}{(1+x)\hat{Z}p_{\min}}\right)^2(1+x)\hat{Z}p_{\min}/2}\\
\end{split}
\end{equation*}
}Using integration by parts and some tricks about inequality, we can show that when $\varepsilon'=\omega(n^{-\varepsilon})$, $E(Y)\leq T+o(1)$. The procedure is regular and tedious, thus we omit it in this work.

So far, we have $E(Y)\leq (1+\varepsilon')\hat{Z}+o(1)$ for any $\varepsilon'=\omega(n^{-\varepsilon})$.
	Combine the fact that $E(OPT)\geq Z_{flow}$ and $\hat{Z}\leq Z_{flow}$, we have $\frac{E(Y)}{E(OPT)}\leq 1+\varepsilon'+o\left(\frac{1}{M}\right)$.
	Recall that to employ the Chernoff Bound, we need $\delta\geq \delta' > 0$. As $\delta' \geq \frac{\varepsilon'}{1+\varepsilon'}-\frac{1}{(M-n)p_{\min}}$, $\delta' > 0$ is satisfied by setting $\varepsilon'$ to be any number greater than $\frac{1}{(M-n)p_{\min}-1}=O(n^{-\varepsilon})$. Thus  setting $\varepsilon'=\Theta\left(n^{-\varepsilon/2}\right)$, we have {\small$$\frac{E(Y)}{E(OPT)}= 1+\Theta\left(n^{-\varepsilon/2}\right)+o\left(\frac{1}{M}\right)=1+O\left(n^{-\varepsilon/2}\right).$$}
\end{proof}

\noindent\textbf{Theorem \ref{theo:robustness}.} \textit{\theorobustness}\\

Before the proof of this theorem, we first  provide several useful definitions and an important fact.

%Let $u_j(\omega, \tau)$ stand for the number of users with type $u_j$ among the first $\tau$ users in a visit sequence $\omega$.
%For any delivery policy $\pi$, denote $\pi(\omega)$ as the user traffics consumed by $\pi$ when the user visiting sequence is $\omega$.
Given a user visiting sequence $\omega$ with length $t$, we construct a new user visiting sequence $\omega'$ by inserting some users into $\omega$ \emph{arbitrarily}. Here \emph{arbitrarily} means the number, positions and types of those inserted users are decided at will.  Suppose the length of $\omega'$ is $t'$ ($t'>t$), we say $\omega'$ is a \emph{$t'$-expansion} of $\omega$. With these definitions, we have the following fact for the flow based delivery policy. %As our flow based delivery policy employs the expected flow to guide the delivery, it is in fact a monotone delivery policy.

\begin{fact}
	\label{sec4:fact:1}
	Given a user visiting sequence $\omega$ with length $t$ and its $t'$-expansion $\omega'$, if all contracts are satisfied by the flow based delivery policy after the $t$-th user visiting on $\omega$, they will also be satisfied after the $t'$-th user visiting on $\omega'$.
\end{fact}

This fact is clear, as the flow based delivery policy employs the edge capacities $C_{i,j}$ to guide the delivery. 

Now we show the \textsc{Theorem} \ref{theo:robustness}.

\begin{proof}
    We first show that
    \begin{equation}
        \label{equ:exp:1}
        E_{\mathcal{D}}(\pi_f(\hat{\mathcal{D}}))\leq (1+\delta)E_{\hat{\mathcal{D}}}(\pi_f(\hat{\mathcal{D}})).
    \end{equation}
	To show this inequality, we construct a new user type distribution $\mathcal{D}'$ by adding a dummy user type $u_{n+1}$ and call those users belong to type $u_{n+1}$ as dummy users. In this distribution,
	{\small
	\begin{equation*}
		p'_j = \left\{
		\begin{aligned}
			&\frac{\hat{p}_j}{1+\delta},& &j\in [n]&\\
			&\frac{\delta}{1+\delta}.& &j = n+1&
		\end{aligned}\right.
	\end{equation*}
}
	%For $1\leq j\leq n$, $p'_j = \frac{\hat{p}_j}{1+\delta}$ and $p'_{n+1} = $. Note that no ad is targeted on $a_{n+1}$.
	Note that $p'_j \leq p_j$ for $j\in [n]$.

	Consider a user visiting sequence $\omega'$ sampled from distribution $\mathcal{D}'$. Suppose all contracts are satisfied by policy $\pi_f(\mathcal{D}')$ when $\tau'$-th user comes. As we can see, there may be some dummy users among those first $\tau'$ users. Intuitively, if we change those dummy users to other types randomly to make up the bias between $p'_j$ and $p_j$ for $j\in [n]$, the number of users consumed will be no more than $\tau'$ according to \textsc{Fact} \ref{sec4:fact:1}. More formally, for any dummy user in $\omega'$, we change her type to $u_j$ with probability $\frac{p_j-p_j'}{p'_{n+1}}$ for $j\in [n]$. Clearly, the obtained user visiting sequence can be viewed as satisfying the distribution $\mathcal{D}$. Based on \textsc{Fact} \ref{sec4:fact:1}, the user traffic consumed on the obtained sequence is no more than $\omega'$. Thus, we have $E_{\mathcal{D}'}(\pi_f(\hat{\mathcal{D}}))\geq E_{\mathcal{D}}(\pi_f(\hat{\mathcal{D}}))$.

    On the other hand, we have $$E_{\mathcal{D}'}(\pi_f(\hat{\mathcal{D}}))=\sum_{j=1}^{n+1} E_{\mathcal{D}'}(\pi_f(\hat{\mathcal{D}}, u_j))$$ where $E_{\mathcal{D}'}(\pi_f(\hat{\mathcal{D}}, u_j))$ is the number of users of type $u_j$ consumed.
	In next step, we aim to show
	{\small
	\begin{equation}
		\label{equ:exp:tmp:1}
		\sum_{j=1}^n E_{\mathcal{D}'}(\pi_f(\hat{\mathcal{D}}, u_j)) = E_{\hat{\mathcal{D}}}(\pi_f(\hat{\mathcal{D}})).
	\end{equation}
}
	That's the expected non-dummy user traffic used under $\hat{\mathcal{D}}$ and $\mathcal{D}'$ are same.
	
	As $p'_j$ is in proportion to $\hat{p}_j$ for $j\in [n]$, the edge capacity $\hat{C}_{i,j}$ calculated based on distribution $\hat{\mathcal{D}}$ is the same with $C'_{i,j}$ calculated based on $\mathcal{D}'$.
	Given a user visiting sequence ${\omega'}$ sampled according to distribution $\mathcal{{D'}}$. Suppose the length of $\omega'$ is $t$ and all contracts are satisfied exactly when $t$-th user in $\omega'$ comes.
	We drop out all dummy users in $\omega'$, and obtain a user sequence $\hat{\omega}$ which is a valid sequence with regard to distribution $\hat{\mathcal{D}}$. We say $\hat{\omega}$ is a dropped sequence of $\omega'$. As the online part of our flow based delivery policy only employs edge capacities to perform deliver, all ad contracts must be satisfied by $\pi(\hat{\mathcal{D}})$ exactly when the last one in $\hat{\omega}$ comes.
	
	Note that there are same numbers of non-dummy users in $\omega'$ and $\hat{\omega}$. That's to say, for any user sequence sampled based on distribution $\mathcal{D}'$, there is a dropped sequence sampled from $\hat{\mathcal{D}}$, and the non-dummy user traffic consumed are same. For $\hat{\omega}$, denote $S(\hat{\omega})$ as the set of all the sequences sampled from $\mathcal{D'}$ and theirs dropped sequences are $\hat{\omega}$. It is clear that $\Pr(\hat{\omega}|\hat{\mathcal{D}}) = \sum_{\omega'\in S(\hat{\omega})}\Pr(\omega'|\mathcal{D})$, as the latter can be viewed as inserting dummy users arbitrarily into the former. Thus the Equation (\ref{equ:exp:tmp:1}) is held.
	
	According to Wald's Identity, we have  $$E_{\mathcal{D}'}(\pi_f(\hat{\mathcal{D}}, u_{n+1}))=p'_{n+1}E_{\mathcal{D}'}(\pi_f(\hat{\mathcal{D}}))=\frac{\delta}{1+\delta}E_{\mathcal{D}'}(\pi_f(\hat{\mathcal{D}})),$$ thus we can conclude that $E_{\mathcal{D}'}(\pi_f(\hat{\mathcal{D}})) = (1+\delta)E_{\hat{\mathcal{D}}}(\pi_f(\hat{\mathcal{D}}))$. Combining above analysis we get the inequality (\ref{equ:exp:1}).

	Denote $Z_{flow}(\mathcal{D})$ as the expected flow with user type distribution $\mathcal{D}$. Similarly we can define $Z_{flow}(\hat{\mathcal{D}})$. As $(1-\delta)p_j\leq \hat{p}_j$, we can get $Z_{flow}(\hat{\mathcal{D}})\leq \frac{Z_{flow}(\mathcal{D})}{1-\delta}$ based on the definition of expected flow. According to \textsc{Theorem} \ref{theo:2}, given $Mp_{\min}=\Theta(n^{\varepsilon})$, we have
    \begin{eqnarray*}
        &Z_{flow}(\hat{\mathcal{D}})\leq E_{\hat{\mathcal{D}}}(\pi_f(\hat{\mathcal{D}}))\leq (1+O(n^{-\varepsilon/2})) Z_{flow}(\hat{\mathcal{D}}),\\
        &Z_{flow}({\mathcal{D}})\leq E_{{\mathcal{D}}}(\pi_f({\mathcal{D}}))\leq (1+O(n^{-\varepsilon/2})) Z_{flow}({\mathcal{D}}).
    \end{eqnarray*}
    It is easy to get that
    \begin{equation}
        \label{equ:exp:2}
        E_{\hat{\mathcal{D}}}(\pi_f(\hat{\mathcal{D}})) \leq \frac{1+O(n^{-\varepsilon/2})}{1-\delta}E_{{\mathcal{D}}}(\pi_f({\mathcal{D}})).
    \end{equation}
    Combining Equation (\ref{equ:exp:1}) and (\ref{equ:exp:2}), we can finish this proof.
\end{proof}
\section{Extensions of the Flow Based Delivery Policy}
\label{sec:name:extensions}
\noindent\textbf{Representativeness.}\ \ Representativeness requires that the exposures to users of an ads should be diverse. More specifically, when all contracts are satisfied, $\frac{U_{ij}}{W_i}$ should be close to $\frac{U_j}{U(\Gamma(a_i))},$ where $U_j$ is the total user used belonging to type $u_j$, $U_{ij}$ is the number of ad $a_i$ displayed to users with type $u_j$ and $U(\Gamma(a_i))$ is the total number of users used with types in $\Gamma(a_i)$.

The flow based delivery policy can be easily modified to achieve representativeness approximately. In the \textsc{Offline part}, we can just set the edge capacity $C_{i,j}$ as {\small$\frac{W_i\cdot p_j}{\sum_{u_l\in \Gamma(a_i)}p_l}$} and needn't calculate the $\hat{Z}$.
%According to the property of max flow, when the \textsc{Offline part} finish, the edge capacity $C_{i,j}$ is exactly $\frac{W_i\cdot p_j}{\sum_{u_l\in \Gamma(a_i)}p_l}$.
We also modify the \textsc{Algorithm} \ref{alg:greedy} by changing the condition in line \ref{line:if} to be
\begin{center}
	\textit{we can't find such an $i'$ or $C_{i',j}=0$}.
\end{center}
With this modified \textsc{Greedy-Delivery-Rule}, when all contracts are satisfied, we have
\begin{equation}
	\label{equ:temp:4}
	\frac{U_{ij}}{W_i} = \frac{C_{i,j}}{W_i} = \frac{p_j}{\sum_{u_l\in \Gamma(a_i)}p_l}
\end{equation}

Suppose that when all contracts are satisfied, there are $Y$ users visiting the publisher platform. Based on the Chernoff Bound, we can prove that for a constant $\delta>0$ and $j\in [m]$,
%\begin{eqnarray*}
	$(1-\delta)Yp_j\leq U_j\leq (1+\delta)Yp_j$
%\end{eqnarray*}
with high probability.
Consequently, for any $\delta>0$, $(a_i, u_j)\in E$,
{\small
\begin{equation*}
	\frac{U_j}{U(\Gamma(a_i))} \in \left[\frac{(1-\delta)p_j}{(1+\delta)\sum_{u_l\in \Gamma(a_i)}p_l},\frac{(1+\delta)p_j}{(1-\delta)\sum_{u_l\in \Gamma(a_i)}p_l}\right]
\end{equation*}
}with high probability. Combining with Equation (\ref{equ:temp:4}), we can achieve an approximation for representativeness.

\vspace{0.2cm}
\noindent\textbf{Smoothness of Delivery.}\ \ Smoothness of delivery requires that the unexposed amount of each ad decreases steadily with time. As mentioned in work \cite{Chen2012Ad}, the perfect smoothness is usually impossible to achieve due to the intersection of different advertisers. We try to make the delivery procedure smooth while keeping the user traffic consumption unchanged. To achieve this objective, we employ a new delivery rule, that is the \textsc{Smooth-Delivery-Rule} shown in \textsc{Algorithm} \ref{alg:smooth}.
As we can see, when a user visiting the publish platform, the \textsc{Smooth-Delivery-Rule} choose an unsatisfied ad $a_{i}$ to minimize the ratio $\frac{\hat{C}_{i,j}}{C_{i,j}}$. As the type of users are drawn from the distribution $\mathcal{D}$, the smoothness is ensured to some extent.
\begin{algorithm}[h!]
	\caption{\label{alg:smooth}\textsc{Smooth-Delivery-Rule}}
	\Indm
	\textsc{Init:} $\forall (a_i, u_j)\in E$, set $\hat{C}_{i,j}\leftarrow C_{i,j}$\\
	\textsc{Smooth-Delivery-Rule}$(u_j)$\\
	\Indp
	Set $i'\gets \arg\min_{a_i\in \Gamma(u_j), W_j>0}\left\{\frac{\hat{C}_{i,j}}{C_{i,j}}\right\}$.\\
	\If{we can't find such an $i'$}
	{
		\textbf{return} $NULL$.
	}
	Set $\hat{C}_{i',j}\gets \hat{C}_{i',j}-1$.\\
	\textbf{return} $i'$.
\end{algorithm}

\vspace{0.2cm}
\noindent\textbf{Multiple Delivery.}\ \ To generate more profit, the publishers usually display multiple ads in a single web page. Now we suppose $k$ ads can be displayed when a user comes. Our flow based delivery policy can be easy extended to cover this setting. In \textsc{Algorithm} \ref{alg:exptectedmatching}, we construct an expected flow $\hat{F}_G(\hat{Z}, k)$ instead of $\hat{F}_G(\hat{Z})$. There are two changes in $\hat{F}_G(\hat{Z}, k)$: capacity of edge $(a_i, u_j)\in E$ is set to $\hat{Z}$ and that of edge $(u_j, t)$ is set to $\ceil{\hat{Z}kp_j}$. The first change ensures that all contracts can be satisfied with constraint that one ad can't be displayed more than once for single user visiting. Through the second change, we somehow copy the user $k$ times when she visits the platform. We also modify the \textsc{Algorithm} \ref{alg:greedy} and let it return $k$ ads whose residual capacity values $C_{i,j}$ are the top $k$ minimum when a user of type $u_j$ comes. With similar analysis, we can show that the performance of the modified flow based delivery policy keeps unchanged for the multiple delivery setting.

\section{Omitted Experiment Results}
The results for 50 input instances with average degree 10 and Random-Normalization are shown in Figure \ref{fig:expo:50}. % can be found in Appendix A.2 of the full version.
Similar experiments results for Gauss-Perturbation are shown in Figure \ref{fig:expo:gauss} and \ref{fig:expo:50:gauss}. 
\begin{figure}[h!]
	\centering
	\begin{subfigure}[b]{\columnwidth}
		\centering\includegraphics[width=3in]{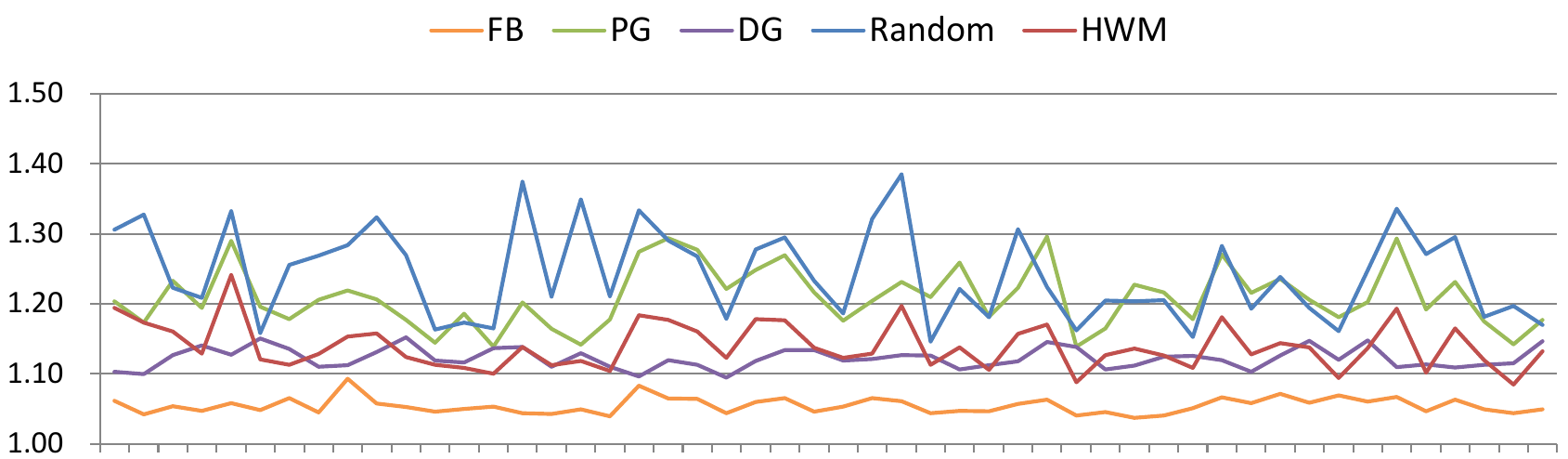}
		\caption{Competitive ratio for 50 input instances}
	\end{subfigure}
	\qquad\\
	\vspace{0.2cm}
	\begin{subfigure}[b]{\columnwidth}
		\centering\includegraphics[width=3in]{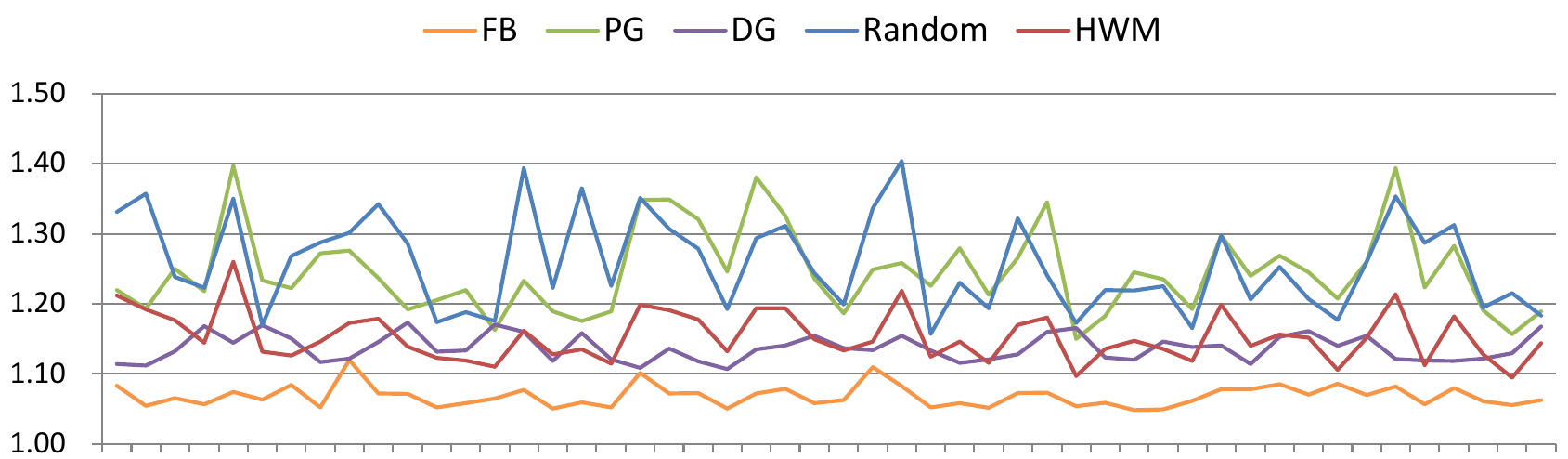}
		\caption{Worst case approximation ratio}
	\end{subfigure}
	\caption{\label{fig:expo:50}Results for 50 input instances with average degree 10 and Random-Normalization. The number of exposures are drawn uniformly from 100 to 5000. The horizontal axis stands for 50 input instances.}
\end{figure}

\begin{figure}%[h!]
	\centering
	\begin{subfigure}[b]{\columnwidth}
		\centering\includegraphics[width=2.8in]{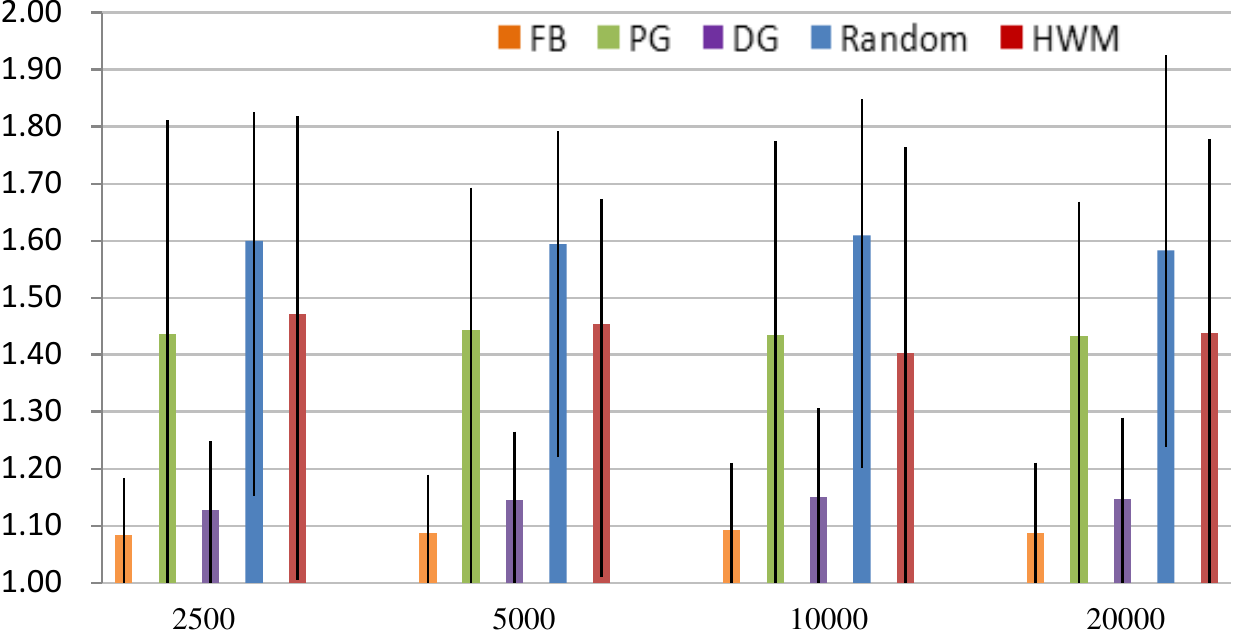}
		\caption{degree $=5$}
	\end{subfigure}
	\begin{subfigure}[b]{\columnwidth}
		\centering\includegraphics[width=2.8in]{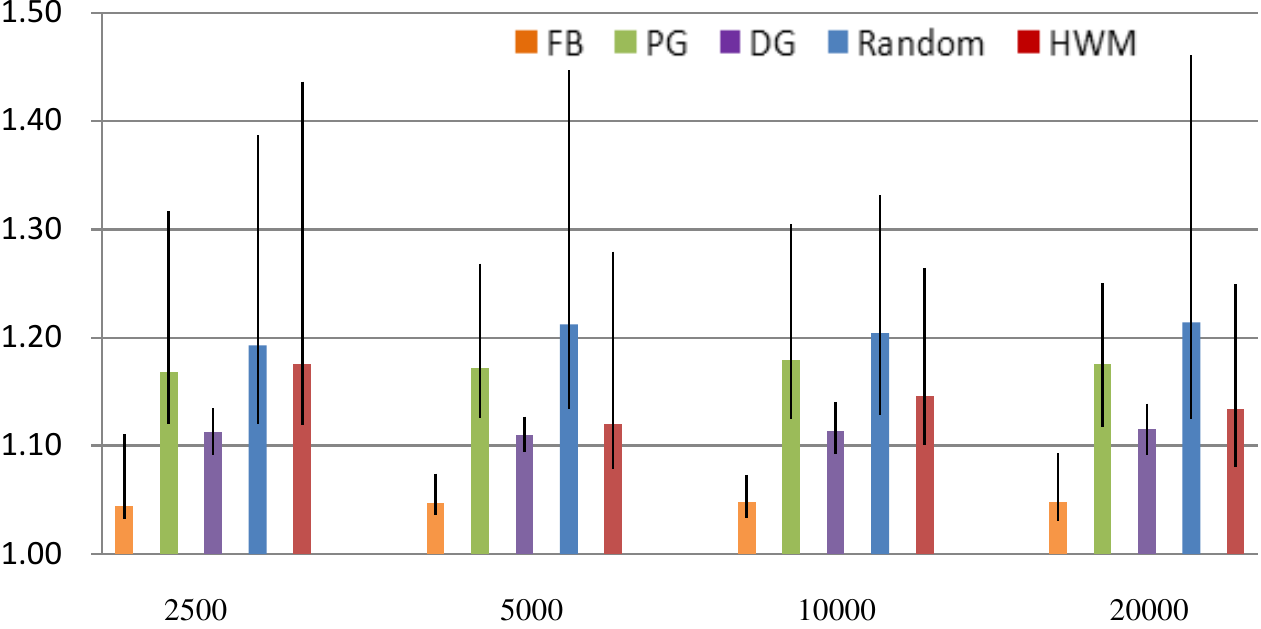}
		\caption{degree $=10$}
	\end{subfigure}
	\caption{\label{fig:expo:gauss}Comparison under different exposures with Gauss-Perturbation.}
\end{figure}

\begin{figure}%[h!]
	\centering
	\begin{subfigure}[b]{\columnwidth}
		\centering\includegraphics[width=3in]{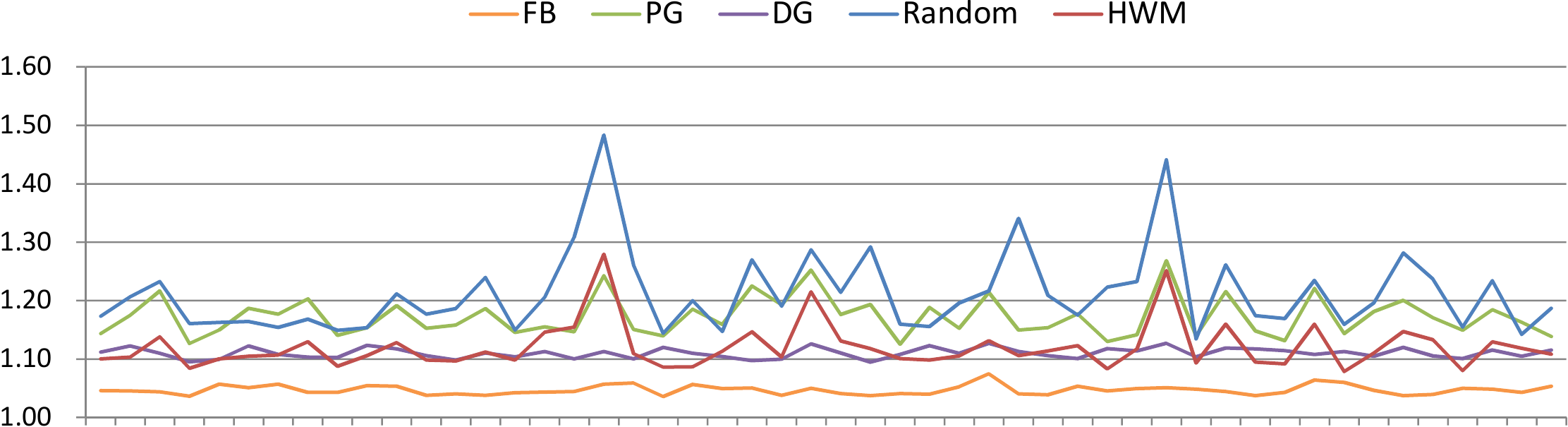}
		\caption{Competitive ratio for 50 input instances}
	\end{subfigure}
	\qquad\\
	\vspace{0.2cm}
	\begin{subfigure}[b]{\columnwidth}
		\centering\includegraphics[width=3in]{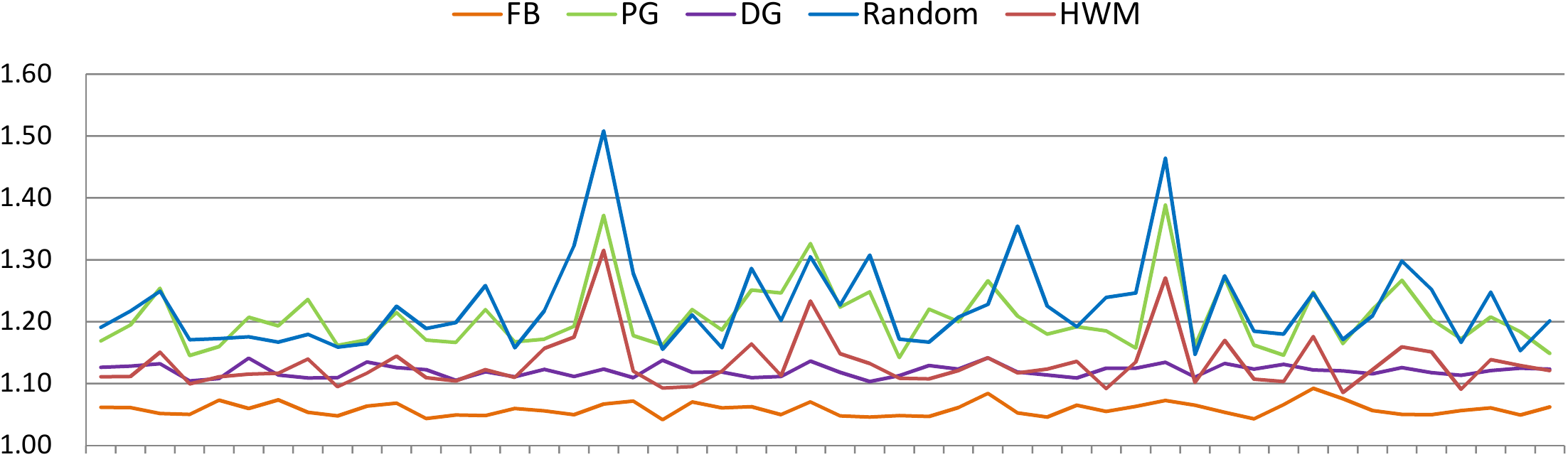}
		\caption{Worst case approximation ratio}
	\end{subfigure}
	\caption{\label{fig:expo:50:gauss}Results for 50 input instances with average degree 10 and Gauss-Normalization. The number of exposures are drawn uniformly from 100 to 5000. The horizontal axis stands for 50 input instances.}
\end{figure}
\end{document}